%% file: main.tex
\documentclass{sig-alternate}

\usepackage{graphicx}
\usepackage{amsfonts}
\usepackage{amsmath}
\usepackage{amssymb}
\usepackage{mathtools}
\usepackage{longtable}
\usepackage{ulem}
\usepackage{url}
\usepackage{ulem}
\usepackage[T1]{fontenc}
\usepackage{times}
\usepackage{graphicx}
\usepackage[bookmarks=true,colorlinks=true]{hyperref}    
\usepackage{array}
\usepackage[boxed,vlined,linesnumbered,noend]{algorithm2e}
\usepackage{tikz}
\usetikzlibrary{matrix,arrows,decorations.pathmorphing, ,decorations.pathreplacing, calc, shapes,patterns, fit, positioning}
\usepackage{subfig} 
\usepackage{float}
\newtheorem{theorem}{Theorem}
\newtheorem{fact}{Fact}
\newtheorem{definition}{Definition}
\newtheorem{remark}{Remark}
\newtheorem{example}{Example}
\newtheorem{lemma}{Lemma}
\newtheorem{observation}{Observation}

\SetKwInOut{Input}{Input}
\SetKwInOut{RestInput}{}
\SetKwInOut{Output}{Output}
\SetKwInOut{RestOutput}{}
\SetKwInOut{Actors}{Actors}
\SetAlFnt{\small}



\input{commands}

\pdfpageattr {/Group << /S /Transparency /I true /CS /DeviceRGB>>} 

\begin{document}


\title{Towards Extending Noiseless Privacy - Dependent Data and More Practical Approach}
\titlenote{First author supported by NCN Research Grant DEC-2013/10/E/ST1/00359. This paper is final version of preliminary paper 'Privacy of Aggregated Data without Noise' and was published on AsiaCCS 2017.}



\numberofauthors{2}
\author{
\alignauthor
Krzysztof Grining \\
       \affaddr{Wroclaw University of Science and Technology}\\
       \affaddr{Faculty of Fundamental Problems of Technology}\\
       \affaddr{Department of Computer Science}\\
       \email{krzysztof.grining@pwr.edu.pl}
\alignauthor
Marek Klonowski \\
       \affaddr{Wroclaw University of Science and Technology}\\
       \affaddr{Faculty of Fundamental Problems of Technology}\\
       \affaddr{Department of Computer Science}\\
       \email{marek.klonowski@pwr.edu.pl}
}

\maketitle

\begin{abstract}
\input{abstract}

\keywords{data aggregation, differential privacy, distributed system }
\end{abstract}

 \input{intro}
 \input{model}

 \input{results}

 \input{conclusion}

%
%
%
\bibliographystyle{abbrv}
\bibliography{bibliography}
\clearpage
\input{appendix}

\end{document}

%% file: commands.tex
\makeatletter
\newlength\min@xx
\newcommand*\xxrightarrow[1]{\begingroup
  \settowidth\min@xx{$\m@th\scriptstyle#1$}
  \@xxrightarrow}
\newcommand*\@xxrightarrow[2][]{
  \sbox8{$\m@th\scriptstyle#1$}  
  \ifdim\wd8>\min@xx \min@xx=\wd8 \fi
  \sbox8{$\m@th\scriptstyle#2$} 
  \ifdim\wd8>\min@xx \min@xx=\wd8 \fi
  \xrightarrow[{\mathmakebox[\min@xx]{\scriptstyle#1}}]
    {\mathmakebox[\min@xx]{\scriptstyle#2}}
  \endgroup}
\newcommand*\xxleftarrow[1]{\begingroup
  \settowidth\min@xx{$\m@th\scriptstyle#1$}
  \@xxleftarrow}
\newcommand*\@xxleftarrow[2][]{
  \sbox8{$\m@th\scriptstyle#1$}  
  \ifdim\wd8>\min@xx \min@xx=\wd8 \fi
  \sbox8{$\m@th\scriptstyle#2$} 
  \ifdim\wd8>\min@xx \min@xx=\wd8 \fi
  \xleftarrow[{\mathmakebox[\min@xx]{\scriptstyle#1}}]
    {\mathmakebox[\min@xx]{\scriptstyle#2}}
  \endgroup}

\usepackage{xcolor}

%% file: abstract.tex

In 2011 Bhaskar et al. pointed out that in many cases one can ensure sufficient level of privacy without 
adding noise by utilizing adversarial uncertainty. Informally speaking, this observation comes from the fact that
if at least a part of the data is randomized from the adversary's
point of view, it can be effectively used for hiding other
values.

So far the approach to that idea in the literature was mostly purely asymptotic, which greatly limited its adaptation in real-life scenarios. In this paper we aim to make the concept of utilizing adversarial uncertainty not only an interesting theoretical idea, but rather a practically useful technique, complementary to differential privacy, which is the state-of-the-art definition of privacy. This requires non-asymptotic privacy guarantees, more realistic approach to the randomness inherently present in the data and to the adversary's knowledge.

In our paper we extend the concept proposed by Bhaskar et al. and present some results for wider class of 
data. In particular we cover the data sets that are dependent. We also introduce rigorous adversarial model. Moreover,
 in contrast to most of previous papers in this field, we give detailed (non-asymptotic) results which is motivated by practical reasons.
Note that it required a modified approach and more subtle mathematical tools, including Stein method which, to the best of our knowledge, was not used in privacy research before. 

Apart from that, we show how to combine adversarial uncertainty with differential privacy approach and explore synergy between them to enhance the privacy parameters already present in the data itself by adding small amount of noise.

%% file: intro.tex
\section{Introduction}\label{sect:intro}

Let us imagine a following problem. There is a set of users and each of them keeps a single value. Analogously, we can think about a database with $n$ records, each corresponding to a specific user. We have to reveal some aggregated statistic (say, the sum of all single values) and preserve the privacy of individuals (say, modeled using standard \textit{differential privacy}  notion). In recent years there have been many very promising results, both for the case where the privacy is governed by a trusted authority (database curator) and for the case where the database is distributed (see for example \cite{PaniShi} and \cite{Rastogi} where the authors use combination of cryptography and privacy preserving techniques). However, the standard differential privacy has an obvious drawback which is a necessity of adding a carefully calibrated noise to the final answer to the query. This approach is not always satisfactory. In some cases we may need to have exact aggregated statistic. 
Moreover, as pointed in some recent papers, adding noise may lead to significant errors in aggregated statistic.
Even if having noisy response is acceptable for a given scenario, the resulting statistics may be too far from the exact values to be usable in practice (see \cite{NaszeACNS,MKAP}). Finally, adding noise, specifically from a non-standard distribution, can be technically problematic -- especially when the aggregated data may come from small, computationally  constrained devices. These facts lead to a somewhat reluctant adaptation of the differential privacy notion in real life applications, despite its undeniable merits.

One may ask if it is possible to circumvent the problem of adding noise while preserving the differential privacy of users. Unfortunately, in the paradigm of standard differential privacy, adding noise is inevitable. Moreover, if we assume that users operate independently and cannot cooperate on adding randomized values used to perturb the original data (which is often the case in distributed systems), the size of aggregated noise has to be $\Omega(\sqrt{n})$, where $n$ is the number of users (as proved in \cite{HubercikBound}).
 
On the other hand, observing some real-life applications of data aggregation one can have an intuition that often it is safe to release aggregated data without adding noise and such act does not expose any individuals' privacy, as pointed out in the seminal paper \cite{prevDP2}. One of classic examples is the average national income. It is clear that such an information says in practice nothing significant about the specific incomes of any of our neighbors, even though they took part in the survey. 
Even revealing the average income of employees in a big company should be secure in terms of privacy of individuals. In contrast, revealing the exact  average income (or maximum income)  in a small community exposes users to obvious risk of privacy breach. 

These intuitions have already been considered in a few papers, namely \cite{prevDP1,prevDP2,prevDP3} to mention the most significant ones, where the authors propose relaxations of the differential privacy model which ''utilizes'' the randomness inherently present in the data itself. Our work can be seen as a continuation and extension of the line of research where the authors leverage adversarial uncertainty. However, in contrast to previous results we focused on detailed, non-asymptotic analysis of the relaxed model, which is motivated by practical needs. Note also, that in the regime of adversarial uncertainty one has to take into account the randomness inherently present in the data, especially the dependencies which naturally appear in real-life scenarios. Therefore, we concentrate also on (locally) \textbf{dependent} data, which importance we justify in Section~\ref{ssect:DEP}. This required using different mathematical tools (e.g. Stein method, see~\cite{fundStein}). To the best of authors knowledge, that type of technical approach was not used previously in the privacy preserving context. Possibly due to the fact that so far the dependent data was not considered in a wide sense in previous papers concerning utilizing adversarial uncertainty.

The intuition behind the \textit{noiseless privacy} approach is that in real life scenarios it might be too pessimistic to assume that the adversary knows almost every record in the database. This assumption seems far too strong, yet it stands at the heart of standard differential privacy. Indeed, it is hard to expect that revealing the exact average worldwide income would in any way harm privacy of any single individual. However, according to differential privacy definition, that would be unacceptable. Intuitively we realize that if an average income (or other value) of a ''large'' set of participants is revealed, there should not be a privacy breach. The authors of \cite{prevDP2} and their notion of noiseless privacy capture that intuition. Their approach allows database designer to check whether the data satisfies desired privacy parameters, and if it does, just reveal the aggregated value without adding any noise. Unfortunately, their results are mostly only asymptotic which makes it hard to use in practice, due to unknown constants which may hide the real size of privacy parameters. Using our methods we give \textbf{explicit bounds} for privacy parameters. From practitioner's point of view, this allows to construct efficient algorithms by directly using our results. Moreover, for the few non-asymptotic results in~\cite{prevDP2} we show that our methods give better bound for privacy parameters. Despite the merits (and theoretical importance) of leveraging adversarial uncertainty, for this approach to become a state-of-the-art privacy promise for various kind of data aggregation problems, it has to be easy to use and quantify for practitioners. Showing precise bounds for privacy parameters and also considering dependent data is the way to make noiseless privacy more useful in practice, which is the purpose of our paper. 

To the best of our knowledge, so far in the privacy literature the idea of combining standard differential privacy techniques (i.e. Laplace mechanism, see \cite{DworkAlgo}) with adversarial uncertainty. Intuitively we can think that in the case where the data has much randomness, we should be able to add smaller noise than in the case where the data is deterministic from the adversary's perspective. Due to our novel approach, we give explicit bounds for privacy parameters which allows us to explore the synergy between differential privacy methods and noiseless privacy approach. We describe and analyse this synergy in Section~\ref{ssect:AddNoise}.


In our paper we follow the model from \cite{prevDP2}, yet present it in a more convenient way for our approach. We show that this definition is coherent with classic (computational) differential privacy -- formally speaking it is an extension. This approach can be seen as utilizing ``uncertainty`` that naturally appears in some data sources to hide the contributions of individuals in the aggregated outcome.   
We depict wide classes of data that can be handled without adding noise and also give the explicit privacy parameters instead of only asymptotic results. Due to explicitly given parameters, our theorems can be seen as ''off the shelf'' ways for a practitioner to check whether he can safely release the data without any noise or not.

\subsection{Our results and organization of this paper}

Our contribution is as follows:
\begin{itemize}
 \item We extend the paradigm of utilizing adversarial uncertainty for the case of dependent data (Theorems~\ref{THMDep} and~\ref{THMAux2}).
\item We explore the synergy between standard differential privacy methods and noiseless privacy approach (Theorem~\ref{THMDosypywanie}).
 \item We propose an adversarial model (Subsection~\ref{ssect:AdvModel}) and explicit procedure for preserving privacy (Figure~\ref{fig:flowchart}), which is easy to use for practitioners.
\item We give improved and explicit (non-asymptotic) bounds for the privacy parameters (Theorems~\ref{ThmInd} and~\ref{THMAux1}).
 \end{itemize}

We believe that this contribution is a step towards more practical constructions of privacy protocols which utilize adversarial uncertainty. Note that, for the first time, we consider wide class of dependent data. Moreover, our results state that the party responsible for privacy does not need to know neither the exact structure of dependencies nor the exact distribution of the data. Upper bounds for the size of the greatest dependent subset and the sum of centralised third moments (or fourth in case of dependent data) are sufficient to use our results in practice. To achieve it, we used different methods than were used in context of adversarial uncertainty before.


The rest of this paper is organized as follows. In Section~\ref{sect:model} we explain the motivations, recall the idea of utilizing adversarial uncertainty from \cite{prevDP2} in a way that is more convenient for presenting our results and provide some formalism that can be seen as an extension of differential privacy notion. We also introduce and discuss our adversarial model and some possible applications.
In the next Sections we present our results. In Section~\ref{ssect:IND}  we focus on the case when 
from the adversary's perspective the aggregated data is a set of independent random values.  Most important is the case discussed in Section~\ref{ssect:DEP}, where we allow the adversary to know \textit{a~priori} some dependencies between data. Note however, that the data owner do not have to know the exact dependencies in the data. Then in Section~\ref{ssect:Aux} we discuss situation where the adversary has an exact knowledge of the values of some subset of data values. Finally in Section~\ref{ssect:AddNoise} we explore the idea of combining adversarial uncertainty with standard differential privacy approach. 

In our paper  we consider privacy guarantees for any fixed size of data, since purely asymptotic approach  seems to be inadequate 
for typical areas of application. Let us stress that we present formulas that can be used for deciding if revealing aggregated data from a given types of data 
is secure even for a moderate number of users. At the end in Section~\ref{sect:prev} we recall some previous and related work. We conclude and outline the future work in Section~\ref{sect:conclusion}. 
Since our paper is quite technical, for the sake of clarity of presentation some of proofs and discussions about the extended definition of privacy have been moved to the Appendix.


%% file: model.tex
\section{Model}\label{sect:model}

As mentioned in the introduction, the main goal of this paper is to make the idea of noiseless privacy (from~\cite{prevDP2}) not only an interesting, theoretical concept, but a practically useful way to guarantee some level of privacy. We want to emphasize that we use the idea (noiseless privacy) from \cite{prevDP2}, yet we present the privacy model in a slightly different way, which seems to be simpler and more convenient for our approach. Moreover it shows direct descendance from classical differential privacy (as presented in~\cite{DworkAlgo}) which may be considered as a special case of the discussed model.

Let us present the aggregation problem in a general way. In the system there are $n$ \textit{users} that may represent different types of parties (organizations, individuals or even sensing devices). Each of them holds a data record $x_i$ (for simplicity we assume that it is a single value). The goal is to aggregate the data and reveal some statistics (say, sum of the values). Note that the database may either be a centralized one, that means there is a database curator whose goal is to reveal the values in a private way (namely via adding some noise to the output), or a distributed one. See that in terms of privacy definition, both these cases are equivalent. They differ in algorithmic approach to these problems. As this paper is about privacy (specifically about utilizing adversarial uncertainty), both these cases are essentially the same for us. Therefore by saying \textit{data} we will mean the set of $n$ values (held either by different parties or by a single curator) which we want to aggregate (i.e. compute the sum of these values) and reveal the obtained statistic to the public. By saying \textit{compromised users} we will mean the subset of data about which the adversary has full knowledge, namely he knows the exact values in this subset. By saying \textit{data owner} we will mean a party that is responsible for preserving privacy of the data by designing an appropriate algorithm, choosing adversarial model parameters (or upper bounds for them) or deciding whether specific privacy parameters are sufficient or if they have to be combined with external noise.
 

\subsection{Modeling privacy of randomized data}
We use a privacy model in which the data (or at least part of it) is considered random from the adversary's perspective, coming from a specific distribution. This kind of approach is quite natural in many scenarios, namely the knowledge of the adversary is usually limited. This ``uncertainty'' can be  utilized.  However, it needs a different definition of privacy than standard differential privacy as in~\cite{DworkAlgo}, because we have to take into account randomized inputs. Following the notion introduced in \cite{prevDP2} we will call this approach \textit{noiseless privacy}. Before we show its formal definition, we need to introduce a following 

\begin{definition}[Adjacent Random Vectors]\label{adjVec}
Let $X = (X_1,\ldots,X_n)$ be an arbitrary random vector and let $X'$ be other random vector. We will say that vectors $X$ and $X'$ are adjacent if and only if 
$$
X' = (X_1,\ldots,X_{i},X_{n+1},X_{i+1},\ldots,X_n),
$$ 
or 
$$
X' = (X_1,\ldots,X_{i-1},X_{i+1},\ldots,X_n),
$$ 
for any $i \in \{1,\ldots,n\}$.
\end{definition}
This essentially captures the notion of data vectors adjacency similar to the one in~\cite{DworkAlgo}, but for random variables rather than deterministic values. See also that if for some deterministic vector $x$ we have $X = x$ with probability $1$, then this definition of adjacency is the same as in~\cite{DworkAlgo}. Note that (as in standard adjacency definition in~\cite{DworkAlgo}) we could as well define adjacency in such a way that instead of adding or removing a vector element, we could simply change its value, this is just the matter of choice and a few straightforward technical changes in proofs. Continuing, we can introduce a following

\begin{definition}[Data sensitivity]
We will say that data vector $X = (X_1,\ldots,X_n)$ and mechanism $M$ have data sensitivity $\Delta$ if an only if
$$
|M(X)-M(X')| \leqslant \Delta,
$$
for every vector $X'$ that is adjacent to $X$.
\end{definition}

Note that this bears close resemblance to the $l_1$-sensitivity defined in~\cite{DworkAlgo}. More detailed comparison of noiseless privacy and standard differential privacy can be found in Appendix.

We can formally define noiseless privacy in the following way
\begin{definition}[Noiseless Privacy]\label{SDP-DEF}
We say that a \\ privacy mechanism $M$ and a random vector $X=(X_1,\ldots,X_n)$ preserve noiseless privacy with parameters $(\epsilon,\delta)$ if for any random vector $X'$ such that $X$ and $X'$ are adjacent we have 
$$
\forall_{B\in \mathcal{B}} P(M(X)\in B) \leqslant e^{\epsilon} P(M(X')\in B) + \delta.
$$
\end{definition}
Intuitively, this definition says that if data can be considered random, then the outcome of the coin flip of any single user does not significantly change the result of \textbf{deterministic} mechanism $M$, whether the user is added to the result, or removed from it. This is very similar to standard differential privacy. A more detailed comparison is moved to the Appendix. Throughout this paper we will use abbreviation $(\epsilon,\delta)$-NP (as in \cite{prevDP2}) to denote noiseless privacy with parameters $\epsilon$ and $\delta$.

Clearly, this model of privacy is a coherent extension of differential privacy. We see it as a generalization of the known differential privacy definition that can be useful for some real life scenarios. See that in Rem.\ref{sect:rem1} (Appendix)  we explained that this model is indeed more general than differential privacy, but if we fix the data as deterministic, it is essentially the same definition. Moreover, in Section~\ref{ssect:AddNoise} we show how the standard differential privacy methods can be combined with noiseless privacy approach.

Whether or not (and to what extent) particular data can be considered random is of course an important problem to be solved by the data holder, and is beyond the scope of this paper. Note that also other papers in this line of research has not yet dealt with this problem which may be a very interesting question for a future work.

See that in noiseless privacy, random data has natural self-hiding properties, even though the mechanisms are deterministic. Instead of relying on the randomness of mechanism (as in the standard differential privacy methods), we can sometimes rely on the inherent randomness of the data itself. Deterministic algorithms have an obvious benefit of not introducing any errors (which are inevitable in standard differential privacy approach due to the addition of noise), so the answer to a query is exact. 

The most common and useful deterministic mechanism would be simply summing all the data. In our paper we explore the privacy parameters of mechanism $M(X) = sum(X)$ for any distribution of the data vector $X$, a wide class of dependencies in the data and the adversarial model defined in Subsection~\ref{ssect:AdvModel}. 

\subsection{Adversarial Model}\label{ssect:AdvModel}

We assume that the adversary:

\begin{itemize}
\item May know the exact data of at most some fraction $\gamma$ of the users.
\item May know the correct distribution (but not the value itself) of the data of the rest of users (note that the distribution for each user might be different).
\item May know the dependencies between some of the data values (if there are any), but only in subsets of size at most $D$.
\end{itemize}

Let us now discuss and justify these assumptions. First of all, one can easily see that in standard differential privacy we essentially assume that the adversary knows the exact data of all users except one. Here we relax this by giving an upper bound on the number of users which are compromised. See that in realistic scenarios it is not very plausible that the adversary indeed knows almost every data record. On the other hand, we still give him quite a lot of power, namely we assume that he knows the distributions of the data, but not the exact values. From the point of view of the adversary, data is a vector of (at least $n-\gamma n$) random variables with known distribution and some known (at most $\gamma n$) data values. See that in sections~\ref{ssect:IND} and~\ref{ssect:DEP} we assume for simplicity that the adversary does not know any exact values (so $\gamma = 0$). We discuss this in Section~\ref{ssect:Aux} where we show how to extend our results for the case where the adversary knows any arbitrary $\gamma n$ exact values.

In real-life data it is quite common to have some dependencies involved. Moreover, the adversary might know about them. To propose a realistic model for noiseless privacy, one has to take it into account. In our model we give the adversary the precise knowledge about all dependencies in subsets of size at most $D$. That essentially means that he does not have an insight into dependencies of subsets of size greater than $D$. Note that it might be the case that such dependencies do not exists (namely the data might really have all dependent subsets of size at most $D$), or simply the adversary does not know about these dependencies and cannot therefore utilize them. Obviously in standard differential privacy notion we do not care about the distribution of data, whether it is dependent or not and so on, which is much easier to comprehend in practical applications. Here, on the other hand, due to the necessity of utilizing the inherent randomness in data instead of adding external noises, we must take such things into account.

See that the is asymmetry between the adversary and other users and even the data owner. Namely we give the adversary power of knowing the exact structure of dependencies (of size at most $D$), while neither users nor the data owner have to know this structure. The necessary parameter to use our results is the upper bound for $D$. Note that the data owner might do some tests for independence of the data (or subsets of the data), i.e. using $\chi^2$-test or other well known statistical methods for testing independence. Information about the upper bound for the size of dependent subsets might also come from strictly engineering knowledge, say due to physical proximity of the subset of sensors or some social knowledge, say subset of users having the same age. This approach to dependencies essentially boils down to the known notion of \textit{dependency neighborhoods} defined as below

\begin{definition}
A collection of random variables $X_1,\ldots,X_n$ has dependency neighborhoods $N_i \subset \{1,\ldots,n\}$, $i \in \{1,\ldots,n\}$ if $i \in N_i$ and $X_i$ is independent of $\{X_j\}_{j\notin N_i}$.
\end{definition}

Observe that the definition of dependency neighborhoods actually says that for specific $X_i$ we know that it is independent of those that are not in its neighborhood. We want to give a general approach to local dependencies scenario, so we do not assume  anything about joint distributions of the dependent subsets. Note that in \cite{prevDP2} the authors gave results for dependent data only for the simplest case of boolean (true/false) data and queries, that is for queries $f$ such that $f: \{0,1\}^n \rightarrow \{0,1\}$. They did not discuss dependencies for more complicated queries and data types. Here, on the other hand, we aim to give a non-asymptotic formula for privacy parameters for \textbf{any distribution} of data and a \textbf{sum query} under dependency regime.

To sum it up, we present a formal definition of adversarial model.

\begin{definition}\label{DefModel}
We will denote a specific instantiation of adversarial model for data vector $X$ by $Adv_X(D,\gamma)$, where
\begin{itemize}
\item $D$ is upper bound for the size of the greatest dependent subset,
\item $\gamma$ is the upper bound for the fraction of the data which values the adversary exactly knows,
\item adversary knows the distribution of data vector $X$.
\end{itemize}
\end{definition}

We believe that while our adversarial model give significantly less power to the adversary than in standard differential privacy notion (which basically gives the adversary almost full knowledge of the data), they still are reasonable and applicable in real-life scenarios. One important remark is that we \textbf{do not} need to predict the exact adversaries knowledge about the dependencies. We only need to know the maximum size of dependency neighborhood, namely the size of largest non-independent subset of data. In fact, we only need an \textbf{upper bound} for that size. Same with the fraction of data values which the adversary knows. To apply our results, which are presented in the next sections, one will also need a lower bound for the variance of data and upper bound for the sum of third and fourth centralized moments for the specific data vector.

\subsection{Applications}

\begin{itemize}
 \item In the case of distributed systems, the users themselves have to secure their privacy using both cryptography and privacy preserving techniques (see for example \cite{PaniShi,Hubercik}). The notion of noiseless privacy and our bounds for privacy parameters are useful especially in distributed case for two reasons. First, in distributed systems quite often the noises which have to be added by users render the data practically useless (too much disturbance). Second, in such systems it is more common to assume that the adversary does not have full knowledge, i.e. can know only at most some fraction of the data. Note also that this paper is solely about privacy and we focus on showing that there are certain data types which does not need any noise added to the final output whether it is a centralized or a distributed case. More details on specific applications for distributed data aggregation can be found in \cite{PaniShi}. See that, if the noiseless privacy assumptions are met and the privacy parameters are satisfying, one could for example run protocols from \cite{PaniShi, Hubercik} with only the cryptographic part, without adding noises to the values. The noises added in standard approach turn out to be quite too big for practical applications in various scenarios (see \cite{NaszeACNS}). 
 \item The idea of noiseless privacy can be used for a wide range of applications including 
networks of sensing environmental parameters, smart metering (e.g, electricity), clinical research, population 
monitoring or cloud services. Most important is however that in all these areas  there are 
natural cases, where we can make some assumptions about the adversaries knowledge.
  \item  Imagine a situation where we have a cloud service which holds shopping preferences of its users. The data is distributed amongst many servers which are completely separated from each other. We assume that at most some (say, 50 percent) of these servers became compromised, which means that at most 50 percent of the values are known to the adversary. Assume that he somehow knows the distribution of the rest (this means that he still has a lot of knowledge about the rest of data) and even some dependencies due to geographical or other reasons. We might know that the greatest dependent subset of our data has size at most $D$ (due to independence tests). This is our model ($Adv_X(D,\gamma)$) for known (or at least upper bounded) $\gamma$, $D$ and distributions of the rest of the data. 
\end{itemize}

%% file: results.tex
\section{Explicit Bounds for \\ Independent Data}\label{ssect:IND}

Assume that we have a database $X$ which consists of $n$ values so $X = \{X_1,\ldots,X_n\}$. Recall that i.i.d. means independent, identically distributed. Let us consider a simple, warm-up scenario, where $X_i$ are i.i.d. random variables and $X_i \sim Bin(1,p)$. We want to aggregate the sum of all these variables so we set $M(X) = \sum_{i=1}^n{X_i} \sim Bin(n,p)$. 

Now we can state a theorem which shows that i.i.d. binomial data has very strong noiseless privacy properties for a wide range of parameters. First we consider the case where $\delta$ is fixed and obtain $\epsilon$ so that the data with summing mechanism is $(\epsilon,\delta)$-NP. Then we fix $\epsilon$ and calculate $\delta$. 
Both cases are considered in the following
\begin{theorem}\label{THMBin}
Let $X=(X_1,\ldots,X_n)$ be a data vector where $X_i \sim Bin(1,p)$ are i.i.d. random variables. If we use mechanism $M(X) =  \sum_{i=1}^{n}(X_i)$ and fix $\delta \geqslant P(M(X)=0)+P(M(X)=n)$, we obtain that it is $(\epsilon,\delta)$-NP for the following
$$
\epsilon =  \begin{cases}
\sqrt{\frac{\ln\left(\frac{2}{\delta}\right)}{2n}}\left(\frac{1}{1-p} - \frac{1}{\sqrt{\frac{\ln\left(\frac{2}{\delta}\right)}{2n}}-p}\right), p \leqslant \frac{1}{2}, \\
\sqrt{\frac{\ln\left(\frac{2}{\delta}\right)}{2n}}\left(\frac{1}{p} - \frac{1}{\sqrt{\frac{\ln\left(\frac{2}{\delta}\right)}{2n}}-(1-p)}\right), p > \frac{1}{2}.
\end{cases}
$$
On the other hand, if $\epsilon > 0$ is fixed, we get
$$
\delta =  \begin{cases}
2\exp\left(-2np^2 \left(\frac{e^{\epsilon} - 1}{e^{\epsilon} + \frac{p}{1-p}} \right)^2\right), p \leqslant \frac{1}{2}, \\
2\exp\left(-2n(1-p)^2 \left(\frac{e^{\epsilon} - 1}{e^{\epsilon} + \frac{1-p}{p}} \right)^2\right), p > \frac{1}{2}.
\end{cases}
$$
\end{theorem}

Proof of this Theorem is quite long and laborious, albeit not very complicated, as it mostly consists of straightforward observations and application of Chernoff bounds. Due to space limitations and mathematical technicalities, the proof has been moved to Appendix.

Let us observe that in Theorem~\ref{THMBin} for constant parameters $p$ and $\delta$ we get  $\epsilon = O\left(\frac{1}{\sqrt{n}}\right)$. It is also worth noting that for $p$ close to $\frac{\lambda}{n}$ or $1-\frac{\lambda}{n}$, $\epsilon$ can be  large, although as long as $p$ is constant, $\epsilon$ still approaches $0$ with $n \rightarrow \infty$.

Similarly,  for $p$ very close to $0$ or $1$ and for small $n$, the value of $\delta$ can be large. Nevertheless  we see that $\delta$ is decreasing \textbf{exponentially} to $0$ with $n \rightarrow \infty$, so for sufficiently large $n$ we still get very small values of $\delta$, even if $p$ was strongly biased. 

One can easily see that this theorem is essentially equivalent to Theorem 5 in~\cite{prevDP2}, but our bounds are tighter and more useful in a practical way, as we give straightforward, non-asymptotic, formulas for $\epsilon$ and $\delta$. On the other hand, authors of~\cite{prevDP2} proved only that due to Chernoff bounds, for a fixed parameter $\epsilon$ the parameter $\delta$ is asymptotically negligible. However, we completed their proof and actually plugged the Chernoff bounds. In Figures~\ref{fig:bin1} and~\ref{fig:bin2} one can see the comparison of our guarantee for parameters, and the guarantee which are given by the (completed) proof in paper~\cite{prevDP2}.

As one can see in Figures~\ref{fig:bin1} and~\ref{fig:bin2}, our Theorem does not only give non-asymptotical, explicit parameters (both for the case where $\epsilon$ is fixed and the case where $\delta$ is fixed), but also, due to slightly more careful reasoning, our bound is tighter than the bound which authors of~\cite{prevDP2} have implicitly shown in their proof.

\begin{figure}[!ht]
    \centering
    \includegraphics[width=0.45\textwidth]{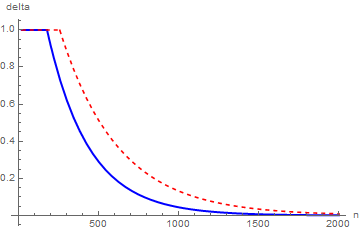}
    \caption{$\epsilon = 0.5$, $p=0.95$, red dashed line is a guarantee for parameter $\delta$ in paper~\cite{prevDP2}, blue thick line is guarantee from our Theorem~\ref{THMBin}.}
    {\label{fig:bin1}}
\end{figure}

\begin{figure}[!ht]
    \centering
    \includegraphics[width=0.45\textwidth]{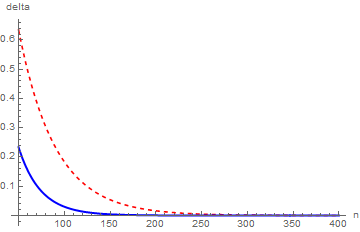}
    \caption{$\epsilon = 1$, $p=0.2$, red dashed line is a guarantee for parameter $\delta$ in paper~\cite{prevDP2}, blue thick line is guarantee from our Theorem~\ref{THMBin}.}
    {\label{fig:bin2}}
\end{figure}

That was just a  warm-up scenario to show how does noiseless privacy work with simple data distribution. Let us move to a more interesting model where users data has different, but still independent distributions. Note that from now on we do not assume any specific distribution of the data.  
Let us recall  two facts. First one is a known result in differential privacy literature.

\begin{fact}[From \cite{DworkAlgo}]\label{factDwork}
Fix $\epsilon > 0$ and $\delta > 0$. Let $c$ such that $c^2 > 2\ln(\frac{1.25}{\delta})$. For random variable $Z \sim \mathcal{N}(0,\sigma^2)$, where $\sigma \geqslant \frac{c\Delta}{\epsilon}$ we have
$$
P[u + Z \in S] \leqslant e^{\epsilon}P[v + Z \in S]+\delta,
$$
where $u$ and $v$ are any real numbers such that $|u - v| \leqslant \Delta$.
\end{fact}

Second fact is a well known theorem in probability theory, one can find it for example in \cite{feller} . 

\begin{fact}[Berry-Esseen Theorem]\label{factBE}
Let $X_1,\ldots,X_n$ be a sequence of independent random variables. Let $EX_i = 0$, $EX_i^2 = \sigma_i^2 > 0$ and $E|X_i|^3 = \rho_i < \infty$. Let $F_n$ denote the cumulative distribution function of their normalized partial sum and $\Phi$ denotes the cumulative distribution function of standard normal distribution. Then
$$
\sup_{x\in\mathbb{R}}|F_n(x) - \Phi(x)| \leqslant \frac{C\cdot \sum_{i=1}^n{\rho_i}}{\left(\sum_{i=1}^n{\sigma_i^2}\right)^{\frac{3}{2}}}
$$
where $C \leqslant 0.5591$.
\end{fact}
The upper bound for constant $C$ comes from \cite{tyurin}.

After stating all necessary facts and definitions, we are ready to present the general theorem for independent data.

\begin{theorem}\label{ThmInd}
Let $X = (X_1,\ldots,X_n)$ be a data vector, where $X_i$ are independent random variables. Let $\mu_i = EX_i$ and $\sigma^2 = \frac{\sum_{i=1}^n{Var(X_i)}}{n}$ and $E|X_i|^3 < \infty$ for every $i \in \{1,\ldots,n\}$. Consider mechanism $M(X) = \sum_{i=1}^{n}(X_i)$. We denote data sensitivity of vector $X$ and mechanism $M$ as $\Delta$. $M(X)$ is $(\epsilon,\delta)$-NP with following parameters
$$
\epsilon = \sqrt{\frac{\Delta^2\ln (n)}{n\sigma^2}},
$$
and
$$
\delta = \frac{1.12 \sum_{i=1}^n{E|X_i-\mu_i|^3}}{\left(n\sigma^2 \right)^{\frac{3}{2}}}(1+e^{\epsilon})+\frac{4}{5\sqrt{n}}.
$$
\end{theorem}

The main idea for proving this theorem is to use Berry-Esseen theorem to deal with random variables of normal distribution instead of the actual distribution of the data. Then we use normal distribution properties to obtain appropriate $\epsilon$ and $\delta$. The proof of this theorem is moved to the Appendix.

See that Theorem~\ref{ThmInd} is essentially a generalization of Theorem 7 in~\cite{prevDP2}, which is a simple consequence of Theorem~\ref{ThmInd}. In our case we give \textbf{explicit formula with all constants}, which asymptotically, after using big oh notation simplifies to the same as in~\cite{prevDP2}. As we emphasized before, explicit formulas for privacy parameters is much more useful for a practitioner than the order of magnitude. Moreover, we do not suffer from limitations of Theorem 7 in \cite{prevDP2}, where the authors assumed that the result of the query has to be $O(\log(n))$. In Section~\ref{ssect:DEP} we also give a generalization for locally dependent data.


Theorem~\ref{ThmInd} gives us very general notion of privacy parameters for summing independent data. Note that in Theorem~\ref{ThmInd} we assumed nothing about the distribution of the data, apart from being independent. The only values we need to know is the variance and sum of appropriate central moments (or upper bounds for these values). We also present an example.

\begin{example}\label{ex2}
\textup{We consider a data vector $X=(X_1,\ldots,X_n)$, where $X_i$ are independent random variables. Let $\Delta = 30$. Let $\sigma^2 = \frac{\sum_{i=1}^n{\sigma_i^2}}{n} = 4$. Let also $\sum_{i=1}^n{E|X_i-\mu_i|^3} = 3\cdot n$. We use mechanism $M(X) =  \sum_{i=1}^{n}(X_i)$. Using Theorem~\ref{ThmInd} we obtain that it is $(\epsilon,\delta)$-NP. Figure~\ref{fig:ind1} shows how the $\epsilon$ decreases with $n$, while Figure~\ref{fig:ind2} shows how $\delta$ decreases with $n$.}
\begin{figure}[h!]
    \centering
    \includegraphics[width=0.45\textwidth]{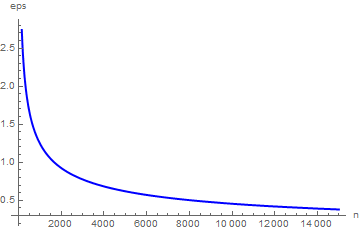}
    \caption{Parameter $\epsilon$ in Example~\ref{ex2}.}
    {\label{fig:ind1}}
\end{figure}

\begin{figure}[h!]
    \centering
    \includegraphics[width=0.45\textwidth]{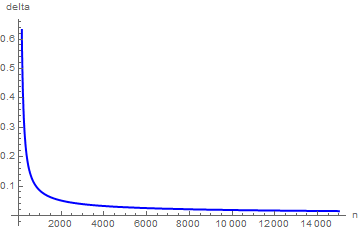}
    \caption{Parameter $\delta$ in Example~\ref{ex2}.}
    {\label{fig:ind2}}
\end{figure}
\textup{We can see that for $n$ around $10000$ parameter $\delta$ is smaller than $0.05$, which is a constant widely used in differential privacy literature, and decreases further. Also, note that for $n \geqslant 10000$  the parameter $\epsilon$ is below $0.5$ which also is a widely used constant in differential privacy papers (see for example~\cite{Hubercik}). Clearly, the parameters keep improving with more users.}
\end{example}

\section{Explicit Bounds for Locally \\ Dependent Data}\label{ssect:DEP}
In the previous section we gave a general treatment for privacy parameters of independent variables. However, in many cases the data has some local dependencies involved. Imagine a situation where we want to collect the data of yearly salary from former students of a specific university. Say, those that finished their education at most 5 years ago. Our goal is to obtain the average yearly salary of all students that finished their education during last five years. Now one can easily see that there will be some local dependencies between the participants as some of the students might work in the same company, launch a startup together or just work in the same field. This will affect their salary and therefore make it locally dependent. Such dependencies are modeled using \textit{dependency neighborhoods} notion, which we defined in Subection~\ref{ssect:AdvModel}.


As previously, we want to take the sum of all our data and show privacy parameters for this mechanism. We are going to take a similar approach as in Theorem~\ref{ThmInd}. That is, we want to bound the distance between the sum of our data and normal distribution. Then, using standard differential privacy properties of normal distribution (described in Fact~\ref{factDwork}) we derive privacy parameters. However, this time we cannot use Berry-Esseen theorem to bound the mentioned distance, as the data is not independent. Instead, we use Stein's method (see for example~\cite{steinIntro,fundStein}), which allows to bound the Kolmogorov distance between two random variables. Apart from that, the presented  reasoning is very similar to Theorem~\ref{ThmInd}. Firstly, we introduce some notation and facts.

\begin{definition}\label{dK}
Let $X$ and $Y$ be a random variables. Let $\mu$ and $\nu$ be their corresponding probability measures. We denote their Kolmogorov distance as $d_K(X,Y)$ which is defined as
$$
d_K(X,Y) = \sup_{t \in \mathbb{R}} \left|F_X(t) - F_Y(t)\right|,
$$
where $F_X(\cdot)$ denotes the cumulative distribution function of $X$. Furthermore, we denote Wasserstein distance as $d_W(X,Y)$ which is defined as
$$
d_W(X,Y) = \sup_{h \in \mathcal{H}} \left|\int{h(x)d\mu(x)} - \int{h(x)d\nu(x)}\right|,
$$
where $\mathcal{H} = \{h: \mathbb{R} \rightarrow \mathbb{R}: |h(x) - h(y)| \leqslant |x-y|\}$.
\end{definition}
These are standard probability metrics, their definition is also given in, for example,~\cite{fundStein}. We also recall a useful relation between Kolmogorov and Wasserstein distance.

\begin{fact}[From~\cite{fundStein}]\label{factStein}
Suppose that a random variable $Y$ has its density bound by some constant $C$. Then for any random variable $X$ we have
$$
d_K(X,Y) \leqslant \sqrt{2C d_W(X,Y)}.
$$
Moreover, if $Y \sim \mathcal{N}(0,1)$, then for any random variable $X$ we have
$$
d_K(X,Y) \leqslant \left(\frac{2}{\pi}\right)^{\frac{1}{4}}\sqrt{d_W(X,Y)}.
$$ 
\end{fact}

Lastly, we recall a theorem from~\cite{fundStein}. 
\begin{fact}[Theorem 3.6 in~\cite{fundStein}]\label{thmStein}
Suppose $X_1,\ldots,X_n$ are \\ random variables such that for every $i$ we have $EX_i^4 < \infty$, $EX_i = 0$, $\sigma^2 = Var[\sum_{i = 1}^n{X_i}]$ and define $W = \frac{\sum_{i=1}^n{X_i}}{\sigma}$. Let the collection $(X_1,\ldots,X_n)$ have dependency neighborhoods $N_i$, $i \in \{1,\ldots,n\}$ and also define $D = \max_{1 \leqslant i \leqslant n}{|N_i|}$. Then, for random variable $Z$ with standard normal distribution we have
$$
d_W(W,Z) \leqslant \frac{D^2}{\sigma^3}\sum_{i=1}^n{E|X_i|^3} + \frac{D^\frac{3}{2} \sqrt{28}}{\sigma^2\sqrt{\pi}} \sqrt{\sum_{i=1}^n{EX_i^4}}.
$$
\end{fact}
This fact is obtained by using Stein's method. We will use these facts to prove a following

\begin{theorem}\label{THMDep}
Let $X=(X_1,\ldots,X_n)$ be a data vector. We consider mechanism $M(X) = \sum_{i=1}^{n}(X_i)$. Let $EX_i = \mu_i$ and $EX_i^4 < \infty$. Suppose there are dependency neighborhoods $N_i$, $i \in \{1,\ldots,n\}$, where $D = \max_{1 \leqslant i \leqslant n}{|N_i|}$. Let $\sigma^2 = Var(M(X))$. If the data sensitivity is $\Delta$ then $M(X)$ is $(\epsilon,\delta)$-NP with following parameters
$$
\epsilon = \sqrt{\frac{\Delta^2 \ln(n)}{\sigma^2}},
$$
and
$$
\delta = c(\epsilon) \sqrt{\frac{D^2}{\sigma^3}\sum_{i=1}^n{E|X^*_i|^3} + \frac{D^\frac{3}{2} \sqrt{26}}{\sigma^2\sqrt{\pi}} \sqrt{\sum_{i=1}^n{E\left(X^*_i\right)^4}}} + \frac{4}{5\sqrt{n}},
$$ 
where $X_i^* = (X_i - \mu_i)$ and
$$
c(\epsilon) = 2(1+e^{\epsilon}) \left(\frac{2}{\pi}\right)^{\frac{1}{4}}.
$$
\end{theorem}
Proof of this theorem is presented in the Appendix. Note that we denote $\sigma^2 = Var(\sum_{i=1}^n{X_i})$ in contrast to $\sigma^2 = \frac{\sum_{i=1}^n{Var(X_i)}}{n}$ as in previous section.

\section{Adversary with Auxiliary \\ Information}\label{ssect:Aux}

So far we have not discussed auxiliary information of the adversary, namely we assumed that the adversary only knows the correct distribution of the data vector. We would like to extend our results from Subsections~\ref{ssect:IND} and~\ref{ssect:DEP} to take into account the adversary's knowledge about the exact values of at most fraction $\gamma$ of users. Let us assume that the auxiliary information of the adversary consists of all records (values) of a subset $\Gamma$ of the data. Let $|\Gamma| = \gamma \cdot n$. Instead of $n$ users contributing to adversarial uncertainty, we will have $(1-\gamma)\cdot n$ users who, due to randomness in their data, make the aggregated value private. This is stated in the following observation

\begin{observation}\label{remarkAux}
Let us consider an adversary with knowledge of exact values of all records of a subset $\Gamma$ of the data. Let $|\Gamma| = \gamma \cdot n$. Then all previous theorems from this paper can be easily adapted to such an adversary by considering data of size $(1-\gamma)n$ instead of $n$ contributing to randomness. This essentially captures the fact that all other users (about whom adversary has no information) still contribute to the randomness of the query. Moreover, if we assume that the adversary has auxiliary information about every record of the data (that is $|\Gamma| = n$) then this model collapses to standard differential privacy, where no uncertainty comes from the data itself. This shows that indeed the standard differential privacy is a special, most pessimistic, case of this model.
\end{observation}

Let us first introduce an extension to Theorem~\ref{ThmInd}, which takes into account the adversary's knowledge about the exact values of fraction of users.

\begin{theorem}\label{THMAux1}
Let $X = (X_1,\ldots,X_n)$ be a data vector, where $X_i$ are independent random variables. Denote set of all indexes by $[n]$. Assume that adversary knows the exact values of at most fraction $\gamma$ of users. Denote the set of indexes of compromised users by $\Gamma$, where $|\Gamma| = \gamma n$. Let $\mu_i = EX_i$ and $\sigma_\Gamma^2 = \frac{\sum_{i \in [n]\setminus\Gamma}{Var(X_i)}}{(1-\gamma)n}$ and $E|X_i|^3 < \infty$ for every $i \in \{1,\ldots,n\}$. Consider mechanism $M(X) = \sum_{i=1}^{n}(X_i)$. We denote data sensitivity of vector $X$ and mechanism $M$ as $\Delta$. $M(X)$ is $(\epsilon,\delta)$-NP with following parameters
$$
\epsilon = \sqrt{\frac{\Delta^2\ln ((1-\gamma)n)}{(1-\gamma)n\sigma_\Gamma^2}},
$$
and
$$
\delta = \frac{1.12 \sum_{i \in [n]\setminus\Gamma}{E|X_i-\mu_i|^3}}{\left(\sum_{i \in [n]\setminus\Gamma}{Var(X_i)}\right)^{\frac{3}{2}}}(1+e^{\epsilon})+\frac{4}{5\sqrt{n}}.
$$
\end{theorem}
\begin{proof}
Proof of this theorem is analogous to proof of Theorem~\ref{ThmInd}, with the single difference that only non-compromised users contribute to the randomness, namely variance of the sum consists of the uncompromised users variance. Therefore when using Berry-Esseen theorem the sum weakly converges to normal distribution with smaller variance than in the case where $\gamma = 0$. Note that in the proof we assume that we know which subset of users is compromised. This might obviously be unknown to the data owner, so we can assume the worst case, namely that the compromised subset $\Gamma$ is the subset of size $\gamma n$ with the greatest variance. This might be checked by the owner (which such subset has the greatest variance) and then the theorem holds, no matter which users are really compromised.
\end{proof}
Similarly we can introduce an extension to Theorem~\ref{THMDep}

\begin{theorem}\label{THMAux2}
Let $X=(X_1,\ldots,X_n)$ be a data vector. Denote set of all indexes by $[n]$. Assume that adversary knows the exact values of at most fraction $\gamma$ of users. Denote the set of indexes of compromised users by $\Gamma$, where $|\Gamma| = \gamma n$. We consider mechanism $M(X) = \sum_{i=1}^{n}(X_i)$. Let $EX_i = \mu_i$ and $EX_i^4 < \infty$. Suppose there are dependency neighborhoods $N_i$, $i \in \{1,\ldots,n\}$, where $X = \max_{1 \leqslant i \leqslant n}{|N_i|}$. Let $\sigma_\Gamma^2 = Var(X \setminus \Gamma)$. If the data sensitivity is $\Delta$ then $M(X)$ is $(\epsilon,\delta)$-NP with following parameters
$$
\epsilon = \sqrt{\frac{\Delta^2 \ln((1-\gamma)n)}{\sigma_\Gamma^2}},
$$
and
$$
\delta = c(\epsilon) \sqrt{\frac{D^2}{\sigma_\Gamma^3}M_X^3 + \frac{D^\frac{3}{2} \sqrt{26}}{\sigma^2\sqrt{\pi}} \sqrt{M_X^4}} + \frac{4}{5\sqrt{(1-\gamma)n}},
$$ 
where 
$$
M_X^3 = \sum_{i \in [n]\setminus\Gamma}{E|X_i - \mu_i|^3}
$$
$$
M_X^4 = \sum_{i \in [n]\setminus\Gamma}{E\left(X_i - \mu_i\right)^4}
$$ 
and
$$
c(\epsilon) = 2(1+e^{\epsilon}) \left(\frac{2}{\pi}\right)^{\frac{1}{4}}.
$$
\end{theorem}
\begin{proof}
Here also the proof is analogous to the proof of Theorem~\ref{THMDep}, and also the difference is that only non-compromised users contribute to the randomness, namely variance of the sum consists of the uncompromised users variance. When we bound the Kolmogorov distance (using Stein method) between the sum and a normal distribution, we use one with smaller variance (namely variance of $X \setminus \Gamma$) than in the case where $\gamma = 0$. As in the previous theorem, a practitioner can assume the worst case, namely that the compromised subset $\Gamma$ is the subset of size $\gamma n$ with the greatest variance.
\end{proof}
These simple extensions of our previous theorem give us a complete insight into noiseless privacy in adversarial model presented in Subsection~\ref{ssect:AdvModel}. The owner of the data (or any party responsible for the privacy in central or distributed database) can give his users a rigorously proved guarantee that as long as at most a fraction $\gamma$ of users is compromised and (in dependent case) if the size of the greatest dependent subset is at most $D$, then the privacy parameters at least as good (we have shown the upper bound for the parameters) as given in Theorem~\ref{THMAux1} if the data is independent or Theorem~\ref{THMAux2} if there are dependencies (known to adversary) in the data.

\section{Synergy between Adversarial \\ Uncertainty and Noise Addition}\label{ssect:AddNoise}

In previous sections we have shown what are the privacy parameters for the randomness inherently present in the data. However, it is easy to imagine that in many cases the amount of randomness (adversarial uncertainty) might be too small to ensure desired size of privacy parameters. Does it mean that in such case we have to step back and use only standard differential privacy methods? Fortunately, it does not. It turns out that the proofs of our theorems are constructed in such a way, that it is possible to extend them to the case where we add some noise to increase the randomness in the data. Even more importantly, it is also easy to quantify how much noise has to be added to improve privacy of the data to the desired parameter in our adversarial model.

To the best of authors knowledge, so far there has not been any approach in the privacy literature to combine the idea of utilizing adversarial uncertainty (randomness in data) and standard approach which is adding appropriately calibrated noise. The idea of adding noise to already somewhat random data is quite simple, yet it needs to be carefully analysed so that one may know exactly how much does it enhance the privacy. It is intuitively very natural to think that the more randomness is present in the data, the less noise (or none, if the randomness itself is enough) we have to add to satisfy desired level of privacy. However, to become a state-of-the-art approach to preserving privacy, this intuition has to be formally introduced, rigorously quantified and proved.

We now introduce a following

\begin{theorem}\label{THMDosypywanie}
Let $X=(X_1,\ldots,X_n)$ be a data vector, the data sensitivity is $\Delta$ and $Var(\sum_{i=1}^n X_i) = \sigma^2$. We consider mechanism $M(X)$ which, due to adversarial uncertainty has certain privacy parameters $(\epsilon_1,\delta)$. We can improve this parameter by adding unbiased noise of variance $\sigma_{\xi}^2$. We show that $M^*(X) = M(X+\xi)$ where $\xi$ is noise (namely random variable such that $E\xi = 0$ and $Var(\xi) = \sigma_\xi^2$) preserves privacy with parameters $(\epsilon,\delta)$, where 
$$
\epsilon = \sqrt{\frac{\Delta^2 \ln(n)}{\sigma^2 + \sigma_\xi^2}}.
$$
\end{theorem}
\begin{proof}
First formula is very easy to obtain from our previous proofs. Similarly as in Theorems~\ref{THMAux1} and~\ref{THMAux2} one can easily see that the sum of data with added noise has variance $\sigma^2 + \sigma_\xi^2$, because the noise is independent from data. Therefore appropriate normal random variables to which we bound the distance of our sum (as in Berry-Esseen theorem and Stein method) will have greater variance, which in turn gives smaller epsilon.
\end{proof}

This approach is quite similar as in the case where the adversary has information about exact values of some fraction of the data, but this time we add variance instead of subtracting it. Improving $\delta$ parameter by adding noise seems to be more difficult, as it might require different approach to previous theorems. We leave it as an interesting problem for future work. After this theorem we can also present an useful observation

\begin{observation}\label{ObsDosypywanie}
We can state Theorem~\ref{THMDosypywanie} in a different way, namely for a fixed privacy parameter $\epsilon$, we obtain that necessary variance of the noise to obtain desired level of privacy is
$$
\sigma_\xi^2 = \max\left(\frac{\Delta^2 \ln(n) - \epsilon^2 \sigma^2}{\epsilon^2},0\right).
$$
\end{observation}
\begin{proof}
This observation is obtained from Theorem~\ref{THMDosypywanie} and quite straightforward algebraic manipulations.
\end{proof}

We also give more specific observation concerning noise having Laplace distribution, which is a common technique in standard differential privacy approach (see for example~\cite{DworkAlgo})

\begin{observation}
Let $X=(X_1,\ldots,X_n)$ be a data vector, the data sensitivity is $\Delta$ and $Var(\sum_{i=1}^n X_i) = \sigma^2$. We consider mechanism $M(X)$ which, due to adversarial uncertainty has certain privacy parameters $(\epsilon_1,\delta)$. We show that $M^*(X) = M(X+\xi)$ where $\xi \sim Lap(\frac{\Delta}{\epsilon_2})$ preserves privacy with parameters $(\epsilon,\delta)$, where 
$$
\epsilon = \sqrt{\frac{\epsilon_1^2 \cdot \epsilon_2^2 \cdot \ln(n)}{2\epsilon_1^2+\epsilon_2^2 \ln(n)}}.
$$
\end{observation}
\begin{proof}
This observation is obtained by straightforward application of Theorem~\ref{THMDosypywanie} for $\xi \sim Lap(\frac{\Delta}{\epsilon_2})$.
\end{proof}

Theorem~\ref{THMDosypywanie} allows the party responsible for preserving privacy to enhance parameter $\epsilon$ of the data itself by using standard methods of differential privacy. See however, that the noise necessary to achieve the desired level of privacy is smaller than using standard differential privacy methods due to the fact, that we already have some level of privacy achieved by the randomness present in the data. We conclude our discussion concerning synergy between adversarial uncertainty and differential privacy approach by showing a following

\begin{example}\label{exDosypywanie}
\textup{We consider a data vector $X=(X_1,\ldots,X_n)$ and mechanism $M(X)$ having the data sensitivity $\Delta = 10$ and $Var(M(X)) = \sigma^2 = \frac{n}{10}$. We enhance the privacy by adding Laplace noise of variance $\sigma_\xi^2$. Using Theorem~\ref{THMDosypywanie} and Observation~\ref{ObsDosypywanie} we can compute what is the necessary variance of noise to obtain privacy parameter $\epsilon = 0.2$ depending on the number of users. See Figure~\ref{fig:addnoise}.}
\begin{figure}[h!]
    \centering
    \includegraphics[width=0.45\textwidth]{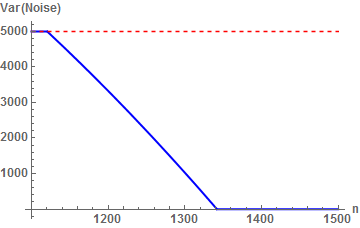}
    \caption{Example~\ref{exDosypywanie}, red dashed line shows the variance of necessary noise for Laplace mechanism using standard differential privacy approach. Blue thick line shows the variance of necessary noise after taking into account the adversarial uncertainty.}
    {\label{fig:addnoise}}
\end{figure}
\textup{See that we have also plotted the variance of noise using differential privacy approach, namely Laplace mechanism (see~\cite{DworkAlgo}). We can see that in this example, for $n$ up to around $1050$ we have to apply standard differential privacy mechanism. Moreover, for $n$ greater than approximately $1350$ we know from our previous results that noise is unnecessary, because the data has sufficient privacy parameters due to inherent randomness. Most interesting, in terms of synergy of adversarial uncertainty and differential privacy methods is the case where $n$ is between $1050$ and $1350$. Here one can see that adding significantly less noise than using standard differential privacy approach is sufficient to obtain desired parameter $\epsilon = 0.2$.}
\end{example}

To sum up all our results, we present a flowchart, which shows on high level of abstraction how should the data owner approach the problem of preserving privacy in a general manner. See Figure~\ref{fig:flowchart}.

\begin{figure}
\centering
\tikzstyle{decision} = [diamond, draw, fill=blue!20, text width=6em, text badly centered, node distance=3cm, inner sep=1pt, rounded corners]
\tikzstyle{block} = [rectangle, draw, fill=blue!20, 
    text width=5em, text centered, rounded corners, minimum height=4em]
\tikzstyle{line} = [draw, -latex']
\tikzstyle{cloud} = [draw, ellipse,fill=red!20, node distance=3cm,
    minimum height=2em]
    
\begin{tikzpicture}[node distance = 3cm, auto]
		\node [cloud] (owner) {data owner};
    \node [decision, below of=owner] (assumptions) {Any assumptions about data/adversary?};
    \node [block, right of=assumptions] (DP) {Use standard differential privacy methods (see~\cite{DworkAlgo})};
    \node [block, below of=assumptions] (init) {Initialize adversarial model for data vector (see Definition~\ref{DefModel})};
    \node [decision, below of=init] (isind) {Is data independent ($D=1$)?};
		\node [block, left of=isind] (ind) {Utilize adversarial uncertainty for independent case (use Theorem~\ref{THMAux1})};
		\node [block, right of=isind] (dep) {Utilize adversarial uncertainty for dependent case (use Theorem~\ref{THMAux2})};
		\node [decision, below of=isind] (issat) {Satisfied with privacy parameters?};
		\node [block, below of=issat] (enhance) {Enhance privacy by adding noise (use Theorem~\ref{THMDosypywanie})};
		\node [block, right of=enhance] (release) {Release aggregated statistic};
    \path [line] (assumptions) -- node {yes} (init);
    \path [line] (init) -- (isind);

		\path [line] (isind) -- node {yes} (ind);
		\path [line] (isind) -- node {no} (dep);
    \path [line,dashed] (owner) -- (assumptions);
    \path [line] (assumptions) -- node {no}(DP);
		\path [line] (ind) |- (issat);
		\path [line] (dep) |- (issat);
		\path [line] (issat) -- node {yes} (release);
		\path [line] (issat) -- node {no} (enhance);
		\path [line] (enhance) -- (release);
\end{tikzpicture}
\caption{A flowchart for privacy preserving in a general way.}\label{fig:flowchart}
\end{figure}
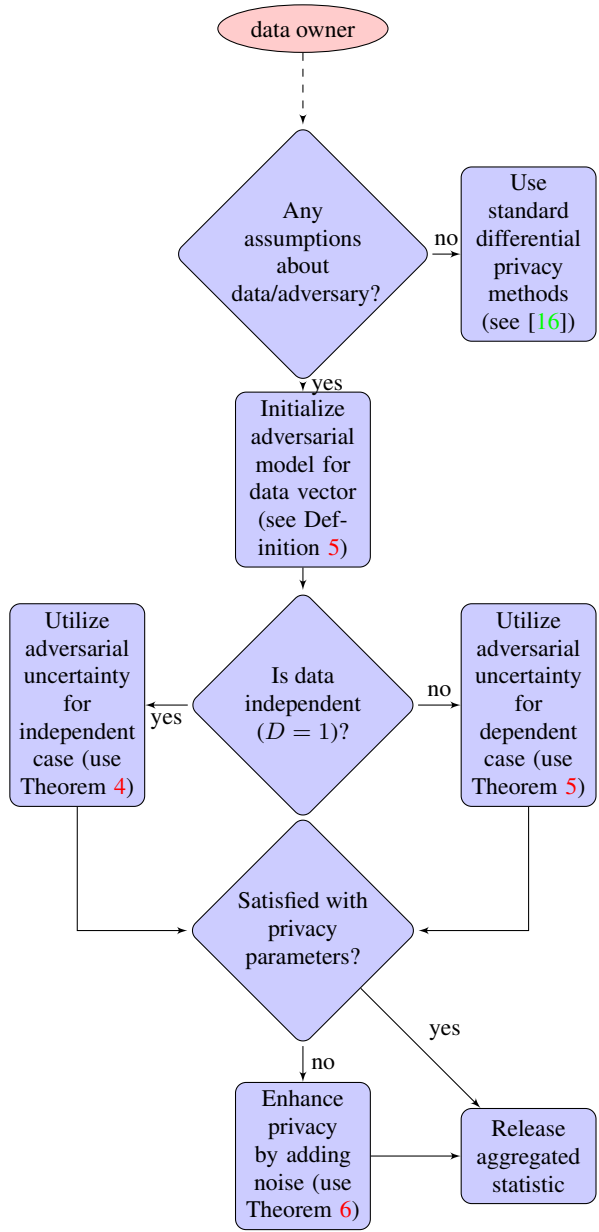

%% file: conclusion.tex
\input{rel}

\section{Conclusions and Further Work}\label{sect:conclusion}

We have shown an explicit bounds for privacy parameters in the case where we can utilize adversarial uncertainty. We have presented specific model of privacy (which boils down to the one given in~\cite{prevDP2}) and introduced model of the adversary. To the best of our knowledge, in the papers concerning leveraging inherent randomness in the data there were only asymptotic results so far. By showing an \textbf{explicit guarantees} for privacy parameters, we have made the whole idea more approachable in practice.

Another important contribution of this paper is approaching \textbf{dependent} data, namely using the notion of dependency neighborhoods. To the best of authors knowledge, such approach has not appeared yet in the literature concerning utilizing adversarial uncertainty to give privacy guarantees. There were some very simple cases, but here we give privacy guarantees for \textbf{any} distribution for a wide class of dependencies. Namely we only need to know the size of the largest dependent subset (or the upper bound for the size)

Moreover, we have shown the parameters regardless of the distribution of the data. The data owner only has to plug the variance of the data (or the lower bound for variance), data sensitivity (which is also necessary in standard differential privacy approach) and appropriate central moments. Then he can give a specific privacy guarantee to its users that as long as at most $\gamma$ is compromised and as long as the greatest dependent subset has size $D$. The simplicity of usage for practitioners was very important in this paper. We want these theorems to be usable not only by the privacy experts, but any specific domain experts, so we have made the theorems sort of 'off-the-shelf' formulas to use.

Furthermore, we have shown how does the standard differential privacy approach combines with the notion of inherent randomness in the data.  It turns out that the intuition that if the data is more 'random', then less noise is necessary to achieve specific privacy parameter. We formalize and quantify the level of privacy enhancement. To the best of our knowledge, such attempt was not presented before in the privacy literature. So far the only attempts were either 'all' (as in standard differential privacy methods) or 'none' (as in for example~\cite{prevDP1,prevDP2,prevDP3}). Here we give the data owner the possibility to maintain a tradeoff between these two approaches.

Some questions are still left unanswered and they might be quite interesting both from practicioner's point of view as well as for the theory. We leave them as a future work.
\begin{itemize}

\item How the database (or distributed system) designer should decide about the level of randomness in the database? In other words, even though in many papers we are given various frameworks to instantiate a specific scenario, how should the practitioner decide which instance to use? Even though we give quite a wide choice for the practitioner (he only needs upper bounds for compromised users, variance and, in dependent case, the size of greatest dependent subset) it still might be cumbersome in some cases. A general method for such a problem would be of great practical value. 

\item We hope to find an even more precise approach to connect the randomness in data with its privacy level. A promising direction is to use notion of  \textit{min entropy} notion (see e.g., \cite{INFORMM}) of data source assuming limited dependencies between values kept by users.  

\item Finding a way to improve also the $\delta$ parameter (we have already shown how to improve $\epsilon$) by adding some noise (albeit less than in standard differential privacy) might be very interesting and useful as well.

\end{itemize}

%% file: rel.tex
\section{Previous and Related Work}\label{sect:prev} 

Our paper can be seen as an extension of the ideas introduced in \cite{prevDP2}. The authors of \cite{prevDP2} proposed a new insight considering relaxation of differential privacy which utilizes the uncertainty of the adversary. This was done in a contrast to standard differential privacy, which assumed that the uncertainty has to be injected by the randomized mechanism. Obviously the notion of differential privacy is quite pessimistic, as we assume that the adversary knows almost everything. In many cases it makes differential privacy unusable in practice. The necessity to add noise to the final output may render the data completely useless. Imagine situation where we want to do a taxation audit. The aggregator collects the amount of taxes paid by the individuals and then publish their sum. After adding a noise, this sum will be different than the tax due, but now we do not know whether it is because of the noise added, or if there is some tax evasion undergoing. Very similar example, and also some other, were given in \cite{prevDP2}. This might be an extreme example, but nevertheless, a big magnitude of noise (say linear of the size of the data itself) would be problematic in most practical situations. One such case is discussed in paper \cite{NaszeACNS}, where the magnitude of noises for practical cases is huge, despite good asymptotic properties of the protocol.

In our paper we use the same model as in \cite{prevDP2}. However, here it is presented in a different way, which is more convenient for our proofs. The results we give are more detailed (non-asymptotic) and easy to use in practice and concern any type of data. To the best of our knowledge, previous work in noiseless privacy and its derivatives or generalizations consisted of asymptotic analysis only. The unknown constants hidden in the big oh notation makes it difficult to construct practical algorithms. Furthermore, we also give results for data with (limited) \textbf{dependencies}, which did not appear in \cite{prevDP2} (apart from simple examples). Moreover, we showed that one can combine noiseless privacy with standard approach, namely adding some noise. It turns out that one can enhance the inherent randomness and reach desired level of privacy with less noise than using standard approach. 


There are many other papers that should be mentioned as a related work. Apart from \cite{prevDP2} there were also very interesting and important papers concerning various approaches to leveraging adversarial uncertainty in privacy, especially \cite{prevDP1,prevDP3}.

Both in \cite{prevDP1} and \cite{prevDP3} the authors proposed a frameworks (called ''coupled-worlds privacy'' and ''Pufferfish'', respectively) for specifying privacy definitions utilizing adversarial uncertainty. They could be instantiated in various ways, one of which boils down to noiseless privacy. These papers are important generalizations of ideas in \cite{prevDP2}, however the main goal of its authors is extending and generalizing privacy definitions. Our paper, on the other hand, focuses on extending the types of data which have good noiseless privacy parameters, on introducing dependencies in the data and combining noiseless privacy with standard approach. Moreover, we focus on detailed results which can be easily applied in real-life scenarios of data aggregation.

Another paper that is somehow related to this one is~\cite{chinczyki}, where the authors utilized sampling to enhance privacy. They have also given non-asymptotic privacy guarantees. However, the authors of~\cite{chinczyki} show how we can get differentially private data using k-anonymity by a simple sampling. On the other hand we consider the problem of aggregation of dependent data. We believe that such approach is more adequate for real-life scenarios. The model we investigated (revealed data) is substantially different. In particular we deal with aggregated data from (possibly) dependent sources. The authors of~\cite{chinczyki} have also proposed a theorem which is essentially very similar to our Theorem~\ref{THMBin}. Note, however, that this theorem is just a toy scenario in our paper, as we focus on any kind of data, not limiting ourselves to specific distribution. Moreover, we introduce local dependencies in the data.

Obviously, our paper is also strongly related to any work concerning data aggregation under differential privacy regime, whether the data is centralized or distributed.

Our results can for example be used in~\cite{PaniShi} 
wherein authors construct a mechanism that allows the untrusted aggregator  to learn only the intended statistics 
but no additional information. Moreover the statistics revealed to the aggregator satisfy differential privacy. 
The result is obtained by combining applied cryptography techniques to hide partial results with regular methods used for privacy preserving for the final result, which can be omitted under noiseless privacy regime, thus not introducing any errors. 

There is a long line of papers concerning similar problems as in \cite{PaniShi}, for example two other notable papers~\cite{Rastogi} and  \cite{elen}. In both of them, the authors  use a substantially different
model of security. Moreover in the latter the users communicate between each other, while in~\cite{PaniShi} as well as in our paper we assume that there is a communication between aggregator and individual users only.  

 Note that most of protocols described in these related papers  fail to provide the correct output even if only a single user abstains from sending his  share of 
the input. The solutions for dynamic networks have been presented in \cite{NaszeACNS} and \cite{Hubercik}.
Approach based on~\cite{PaniShi} was also focused on more advanced particular processing of aggregated data  (e.g., evaluation and  monetization) while keeping privacy of users 
is discussed in several papers (\cite{PS2x,PS3x,Gist,MKAP}). Another vain of protocols represent~\cite{PDA,PDA2} wherein authors present some aggregation methods that preserve privacy, however they do not consider dynamic changes inside of the network. The latter also considers data poisoning attacks, however the authors do not provide rigid proofs. In~\cite{6171193,jajodia} the authors present a framework for some aggregation functions and consider the confidentiality of the result, but leaving nodes' privacy out of scope. Clearly there are many papers discussing aggregation protocols without considering security nor privacy issues (e.g.,~\cite{unsecureAgg,tag}). There is a long list of papers devoted to fault tolerant aggregation protocols (\cite{Feng2011451,Jhumka20141789,4147120}) for significantly different settings. 

One could use the notion of noiseless privacy, especially the explicit results given in our paper, to get rid of the noise addition (thus, the error introduced in result of a query) in many protocols in papers mentioned in this section.

As a related work we shall point also a huge body  of papers dealing with differential privacy notions and their extension. The idea of differential privacy has been introduced 
for the first time in~\cite{dW1}, however its precise formulation in the widely used form  appeared in~\cite{dW2}. Most important properties have been introduced in papers~\cite{dW3,dW4}. There is a long list of papers that can be seen as a direct extension of~\cite{dW1} i.e., \cite{Gist,dW3}. In all that papers a substantially different trust model is used.  Namely there is a party called \textit{curator} that is entitled to see all participants' data in the clear and releases the computed data to wider (possibly untrusted) audience. 

Paper~\cite{Niss} presents aggregation of elements of dataset from  perspective of preserving differential privacy. The presented framework significantly differs from our approach in a few points. First of all, it uses adding noise to raw data.

An introduction to differential privacy  can be  found in~\cite{DIFFPR}. An excellent, comprehensive  description  of recent results can be found in~\cite{DworkAlgo}. 


%% file: appendix.tex
\appendix \label{sect:App}

\section{Technical proofs}

\subsection{Proof of Theorem 1}
\begin{proof}
First, we will prove the following lemma.

\begin{lemma}\label{BinLemma}
Let $X \sim Bin(n,p)$. Fix an arbitrary $\lambda > 0$ such that $(np - \lambda) > 0$ and $(np+\lambda) < n$. Let $u \in [np-\lambda,np+\lambda]  \cap \mathbb{Z}$ and let $v \in \mathbb{Z}$ such that $|u - v| = 1$. We have
$$
P(X = u) \leqslant e^{\epsilon}P(X = v),
$$
where 
$$
\epsilon = \begin{cases}
\frac{\lambda}{n}\left(\frac{1}{1-p} - \frac{1}{\sqrt{\frac{\lambda}{n}}-p}\right), p \leqslant \frac{1}{2}, \\
\frac{\lambda}{n}\left(\frac{1}{p} - \frac{1}{\sqrt{\frac{\lambda}{n}}-(1-p)}\right), p > \frac{1}{2}.
\end{cases}
$$
\end{lemma}
\begin{proof}
We want to bound $\frac{P(X=u)}{P(X=v)}$, where $|u - v| = 1$ and $X \sim Bin(n,p)$. Furthermore, we know that $u \in [np-\lambda,np+\lambda]  \cap \mathbb{Z}$. First observe that we get the biggest ratio either for the smallest or greatest possible $u$.  Moreover, if $p \leqslant \frac{1}{2}$ we get the biggest ratio for  the smallest possible $u$.
Therefore it remains to check these two cases, calculate $\epsilon_1$ and $\epsilon_2$  and pick $\epsilon = \max(\epsilon_1,\epsilon_2)$. 

Let us begin with the case where $p \leqslant \frac{1}{2}$. Then we have $X \sim Bin(n,p)$. One can easily check that the greatest possible ratio is for $u = \lceil np-\lambda \rceil$ and $v = (u-1)$. We can bound it in the following way
\begin{align*}
\frac{P(X = \lceil np-\lambda \rceil )}{P(X = \lceil np-\lambda \rceil-1)} &= \frac{n- \lceil np-\lambda \rceil}{\lceil np-\lambda \rceil} \cdot \frac{p}{1-p} \leqslant \\
&\leqslant \frac{n-np+\lambda}{np-\lambda} \cdot \frac{p}{1-p} = \\
&= \frac{1+\frac{\lambda}{n(1-p)}}{1-\frac{\lambda}{np}} \leqslant \frac{\exp(\frac{\lambda}{n(1-p)})}{1-\frac{\lambda}{np}}.
\end{align*}
Ultimately we are interested in the natural logarithm of that ratio. We have
\begin{align*}
\epsilon_1 &= \log\left(\frac{\exp(\frac{\lambda}{n(1-p)})}{1-\frac{\lambda}{np}}\right) = \frac{\lambda}{n(1-p)} - \log\left(1-\frac{\lambda}{np}\right) \leqslant \\
&\leqslant \frac{\lambda}{n(1-p)} - 1 + \frac{1}{1-\frac{\lambda}{np}} = \lambda \left(\frac{1}{n(1-p)} + \frac{1}{np-\lambda} \right) = \\
&= \frac{\lambda}{n}\left(\frac{1}{1-p} - \frac{1}{\frac{\lambda}{n}-p} \right),
\end{align*}
where the inequality comes from the fact that $(1-\frac{1}{x}) \leqslant \log(x)$ for $x > 0$. See also that $1-\frac{\lambda}{np} > 0$, because we assumed that $(np-\lambda) > 0$. We also have $p > \frac{\lambda}{n}$ so all performed derivations are correct. Note that we picked the biggest possible ratio, so for $p \leqslant \frac{1}{2}$ it is true for every $u \in [np-\lambda,np+\lambda]  \cap \mathbb{Z}$ that
$$
\frac{P(X=u)}{P(X=v)} \leqslant e^{\epsilon_1} \iff P(X=u) \leqslant e^{\epsilon_1}P(X=v),
$$
where $|u-v| = 1$. Now let us assume that $p > \frac{1}{2}$. In that case the greatest possible ratio is for $u = (np+\lambda)$ and $v = (u+1)$. One can easily see, that we can simply consider $Bin(n,1-p)$ and apply exactly the same reasoning as before. That leaves us with
$$
\epsilon_2 = \frac{\lambda}{n}\left(\frac{1}{p} - \frac{1}{\frac{\lambda}{n}-(1-p)} \right).
$$
Similarly, we have $(1-p) > \frac{\lambda}{n}$, so there is no division by $0$. In the end, we conclude that for a fixed $\lambda$ we have the following:
$$
\epsilon = \begin{cases}
\frac{\lambda}{n}\left(\frac{1}{1-p} - \frac{1}{\sqrt{\frac{\lambda}{n}}-p}\right), p \leqslant \frac{1}{2}, \\
\frac{\lambda}{n}\left(\frac{1}{p} - \frac{1}{\sqrt{\frac{\lambda}{n}}-(1-p)}\right), p > \frac{1}{2}.
\end{cases}
$$
In the end we found $\epsilon$, which has a property that for  all \\ $u \in [np-\lambda,np+\lambda]  \cap \mathbb{Z}$ and $|u - v| = 1$ it holds that
$$
P(X = u) \leqslant e^{\epsilon}P(X = v),
$$
which concludes the proof of this lemma.
\end{proof}

Now we can continue with the proof of our Theorem. Let us begin with the first case, where $\delta$ is fixed. One obvious observation is that $M(X) \sim Bin(n,p)$. Using Chernoff bounds (see for example~\cite{concentration}) for binomial distribution we get
$$
P(M(X)\geqslant np +\lambda) + P(M(X)\leqslant np - \lambda) \leqslant 2\exp\left(-\frac{2\lambda^2}{n}\right).
$$
We want to limit the tail probability by parameter $\delta$, so we want to find a $\lambda$ such that the right side of this inequality is equal to $\delta$. This yields
$$
2\exp\left(-\frac{2\lambda^2}{n}\right) = \delta \iff \lambda = \sqrt{\frac{n\ln{\frac{2}{\delta}}}{2}}.
$$
Let us denote the set $S = \{\lceil\mu-\lambda\rceil,\ldots,\lfloor\mu+\lambda\rfloor\}$, which is exactly the support of $M(X)$ without the tails which probability we just limited by $\delta$. Now we have to find $\epsilon$ such that, apart from the tails, the following condition is satisfied
$$
\forall_{B\subset S}\left(\left|\log\left(\frac{P(M(X)\in B)}{P(M(X')\in B)}\right)\right| \leqslant \epsilon \right).
$$
It is easy to see that instead of checking all subsets of $S$, we can check only the single values, because taking a single value with a bigger ratio yields worst case bound.
For that, we can use Lemma~\ref{BinLemma}. We indeed have $M(X) \sim Bin(n,p)$. Recall that we assumed $\delta \geqslant P(M(X)=0)+P(M(X)=n)$. This means that at least $0$ and $n$ are in the tail that we already limited by $\delta$. Therefore, $(np - \lambda) > 0$ and $(np+\lambda) < n$. Applying Lemma~\ref{BinLemma} for $M(X)$ and $\lambda$ we obtain that
$$
P(M(X) = u) \leqslant e^{\epsilon}P(M(X)=v),
$$
for $u \in S$ and $|u - v| \leqslant 1$. Observe that $\frac{\lambda}{n} = \sqrt{\frac{\ln\left({\frac{2}{\delta}}\right)}{2n}}$ so from Lemma~\ref{BinLemma} we have
$$
\epsilon = \epsilon(n,p,\delta) = \begin{cases}
\sqrt{\frac{\ln\left(\frac{2}{\delta}\right)}{2n}}\left(\frac{1}{1-p} - \frac{1}{\sqrt{\frac{\ln\left(\frac{2}{\delta}\right)}{2n}}-p}\right), p \leqslant \frac{1}{2}, \\
\sqrt{\frac{\ln\left(\frac{2}{\delta}\right)}{2n}}\left(\frac{1}{p} - \frac{1}{\sqrt{\frac{\ln\left(\frac{2}{\delta}\right)}{2n}}-(1-p)}\right), p > \frac{1}{2}.
\end{cases}
$$
Now see that in our case, for $X_i \sim Bin(1,p)$ i.i.d. we have data sensitivity $1$. One can easily see that adding or removing a single data point can change the sum only by $1$. Therefore we have
$$
P(M(X) \in S) \leqslant e^{\epsilon}P(M(X') \in S)+\delta,
$$
where $X$ and $X'$ are adjacent vectors and $\epsilon = \epsilon(n,p,\delta)$. The addition of $\delta$ comes from the fact that we bound the tails of $M(X)$.

Now we assume that we have a fixed $\epsilon > 0$. Let $\alpha = e^{\epsilon}$ and $w = \frac{p}{1-p}$. We use similar reasoning as in Lemma~\ref{BinLemma}. First let us consider $p \leqslant \frac{1}{2}$. We are interested in the greatest integer $k$ smaller than $np$, which does \textbf{not} satisfy the following
$$
\frac{P(M(X) = k)}{P(M(X) = k-1)} \leqslant \alpha.
$$
We have
$$
\frac{P(M(X) = k)}{P(M(X) = k-1)} = \frac{n-k}{k} \cdot w > \alpha \iff k < \frac{nw}{\alpha + w}.
$$
Now let us pick $\lambda_k = \mu - k > \mu - \frac{nw}{\alpha+w}$. We will bound the tail using Chernoff bound
\begin{align*}
P&(M(X) \leqslant \mu - \lambda_k) \leqslant \exp\left(\frac{-2\lambda_k^2}{n}\right) < \\
&< \exp\left(\frac{-2\left( \mu- \frac{nw}{\alpha+w} \right)^2 }{n}\right) = \\
&= \exp\left(-2np^2 \left(\frac{\alpha - 1}{\alpha + w} \right)^2\right).
\end{align*}
Now we can pick $\delta_1$ in the following way
\begin{align*}
\delta_1 &= P(M(X) \leqslant \mu - \lambda_k) + P(M(X) \geqslant \mu + \lambda_k) \leqslant \\
&\leqslant 2\exp\left(-2np^2 \left(\frac{e^{\epsilon} - 1}{e^{\epsilon} + \frac{p}{1-p}} \right)^2\right).
\end{align*}
When $p > \frac{1}{2}$ we can do similar symmetric reasoning as before, we obtain
$$
\delta_2 \leqslant 2\exp\left(-2n(1-p)^2 \left(\frac{e^{\epsilon} - 1}{e^{\epsilon} + \frac{1-p}{p}} \right)^2\right).
$$
Now we pick $\delta$ which is $\max(\delta_1,\delta_2)$, so we have
$$
\delta = \begin{cases}
2\exp\left(-2np^2 \left(\frac{e^{\epsilon} - 1}{e^{\epsilon} + \frac{p}{1-p}} \right)^2\right), p \leqslant \frac{1}{2}, \\
2\exp\left(-2n(1-p)^2 \left(\frac{e^{\epsilon} - 1}{e^{\epsilon} + \frac{1-p}{p}} \right)^2\right), p > \frac{1}{2}.
\end{cases}
$$
This concludes the proof, because we have found a bound for the subset of possible values which did not satisfy our required ratio. In the end we have
$$
P(M(X) \in S) \leqslant e^{\epsilon}P(M(X') \in S)+\delta,
$$
which concludes the proof.
\end{proof}
\subsection{Proof of Theorem 2}

\begin{proof}
To prove this theorem, we will use Facts \ref{factDwork} and \ref{factBE} from Section~\ref{ssect:IND}. Let $X = \sum_{i=1}^n{X_i}$ and $\sigma^2 = \frac{\sum_{i=1}^n{\sigma_i^2}}{n}$. Let $u,v \in supp(X)$ and $|u - v| \leqslant \Delta$. For any Borel set $B$ let us denote $B_u = \{b + u: b \in B\}$. For simplicity let us, for now, assume that $EX_i = 0$ for every $i$. From assumptions we also know that $E|X_i|^3 < \infty$ for every $i$, so we can use Fact~\ref{factBE}. Let $Z \sim \mathcal{N}(0,n\sigma^2)$. For every $B_u$ we have

$$
P\left(X \in B_u\right) \leqslant P\left(Z \in B_u\right) + 2\delta_1,
$$
where $\delta_1 \leqslant \frac{0.56 \sum_{i=1}^n{E|X_i|^3}}{\left(\sum_{i=1}^n{\sigma_i^2}\right)^{\frac{3}{2}}}$ is the rate of convergence described in Fact~\ref{factBE}. Now we can use Fact~\ref{factDwork}:
$$
P\left(Z \in B_u\right) + 2\delta_1 \leqslant e^{\epsilon}P\left(Z \in B_v\right) +2\delta_1 + \delta_2.
$$
Both $\epsilon$ and $\delta_2$ are parameters from Fact~\ref{factDwork} for the normal distribution with variance $n\sigma^2$ and in case where $|u-v| \leqslant \Delta$. In particular, we can fix $\delta_2 = \frac{4}{5\sqrt{n}}$. From Fact~\ref{factDwork} we get
$$
\epsilon = \sqrt{\frac{\Delta^2\ln (n)}{n\sigma^2}}.
$$
Now we have to return to our initial distribution. Again, we use Fact~\ref{factBE}.
\begin{align*}
&e^{\epsilon}P\left(Z \in B_v\right) +2\delta_1 + \delta_2 \leqslant \\
&\leqslant e^{\epsilon}P\left(X \in B_v\right) +2\delta_1(1+e^{\epsilon}) + \delta_2.
\end{align*}
During this reasoning we already obtained $\epsilon$. We also have 
$$
\delta = 2\delta_1(1+e^{\epsilon}) + \delta_2 \leqslant \frac{1.12 \sum_{i=1}^n{E|X_i|^3}}{\left(\sum_{i=1}^n{\sigma_i^2}\right)^{\frac{3}{2}}}(1+e^{\epsilon})+\frac{4}{5\sqrt{n}}.
$$ 
Note that for simplicity we assumed $EX_i = 0$. One can easily see that for $Y_i = (X_i - \mu_i)$, where $\mu_i = EX_i$ the proof is still correct. Therefore we have
$$
\delta = 2\delta_1(1+e^{\epsilon}) + \delta_2 \leqslant \frac{1.12 \sum_{i=1}^n{E|X_i-\mu_i|^3}}{\left(\sum_{i=1}^n{\sigma_i^2}\right)^{\frac{3}{2}}}(1+e^{\epsilon})+\frac{4}{5\sqrt{n}}
$$
Finally we have
\begin{align*}
P(X \in B_u) &\leqslant e^{\epsilon}P(X \in B_v) + \delta_1(1+e^{\epsilon}) + \delta_2 \leqslant \\
&\leqslant e^{\epsilon}P(X \in B_v) + \delta,  
\end{align*}
which concludes the proof.
\end{proof}

\subsection{Proof of Theorem 3}
\begin{proof}
To prove this lemma, we use facts stated in Section~\ref{ssect:DEP}, namely Fact \ref{factStein} and \ref{thmStein}. We also use Kolmogorov and Wasserstein distances, which were defined in Section \ref{ssect:DEP} in Definition \ref{dK}. We have $X = \sum_{i=1}^n{X_i}$ and $\sigma^2 = Var(X)$. Let $u,v \in supp(X)$ and $|u -v| \leqslant \Delta$. For any Borel set $B$ let us denote $B_u = \{b + u: b \in B\}$. Moreover, throughout the proof we denote $\frac{B_u}{\sigma} = \{\frac{b}{\sigma}: b \in B_u \}$. For simplicity let us, for now, assume that $EX_i = 0$ for every $i$. Let $Z \sim \mathcal{N}(0,n\sigma^2)$. For every $B_u$ we have
$$
P\left(X \in B_u\right) = P\left(\frac{X}{\sigma} \in \frac{B_u}{\sigma}\right).
$$
Recall that we assumed $EX_i = 0$ and $EX_i^4 < \infty$.  Now let $Z \sim \mathcal{N}(0,1)$. From Fact~\ref{thmStein} we have
$$
d_W\left(\frac{X}{\sigma},Z\right) \leqslant \frac{D^2}{\sigma^3}\sum_{i=1}^n{E|X_i|^3} + \frac{D^\frac{3}{2} \sqrt{26}}{\sigma^2\sqrt{\pi}} \sqrt{\sum_{i=1}^n{EX_i^4}}.
$$
Note that for simplicity we assumed $EX_i = 0$. One can easily see that for $X^*_i = (X_i - \mu_i)$, where $\mu_i = EX_i$ the proof is still correct. We have
$$
d_W\left(\frac{X}{\sigma},Z\right) \leqslant \frac{D^2}{\sigma^3}\sum_{i=1}^n{E|X^*_i|^3} + \frac{D^\frac{3}{2} \sqrt{26}}{\sigma^2\sqrt{\pi}} \sqrt{\sum_{i=1}^n{E\left(X^*_i\right)^4}}.
$$

We can use Fact~\ref{factStein} to get Kolmogorov distance of $\frac{X}{\sigma}$ and $Z$. Namely
$$
d_K\left(\frac{X}{\sigma},Z\right) \leqslant \left(\frac{2}{\pi}\right)^{\frac{1}{4}}\sqrt{d_W\left(\frac{X}{\sigma},Z\right)}.
$$
Having Kolmogorov distance of $\frac{X}{\sigma}$ and $Z$, we can proceed further
\begin{align*}
&P\left(\frac{X}{\sigma} \in \frac{B_u}{\sigma}\right) \leqslant \\  
&\leqslant P\left(Z \in \frac{B_u}{\sigma}\right) + 2d_K\left(\frac{X}{\sigma},Z\right) = \\
&= P\left(Z\cdot \sigma \in B_u\right) + 2d_K\left(\frac{X}{\sigma},Z\right).
\end{align*}
Now we can use the property of the normal distribution stated in Fact~\ref{factDwork}.
\begin{align*}
&P\left(Z \cdot \sigma \in B_u\right) + 2\delta_1 \leqslant \\
&\leqslant e^{\epsilon}P\left(Z \cdot \sigma \in B_v\right) +2d_K\left(\frac{X}{\sigma},Z\right) + \delta_1.
\end{align*}
Both $\epsilon$ and $\delta_1$ are parameters from Fact~\ref{factDwork}, for the normal distribution with variance $\sigma^2$ and $|u-v| \leqslant \Delta$. In particular, we can fix $\delta_1 = \frac{4}{5\sqrt{n}}$. From Fact~\ref{factDwork} we get
$$
\epsilon = \sqrt{\frac{\Delta^2 \ln(n)}{\sigma^2}}.
$$
Now we have to return to our initial distribution. Again, we use Facts \ref{factStein} and \ref{thmStein}.
\begin{align*}
&e^{\epsilon}P\left(Z \cdot \sigma \in B_v\right) +2d_K\left(\frac{X}{\sigma},Z\right) + \delta_1 \leqslant \\
&\leqslant e^{\epsilon}P\left(\frac{X}{\sigma} \in \frac{B_v}{\sigma}\right) +2d_K\left(\frac{X}{\sigma},Z\right)\left(1+e^{\epsilon}\right) + \delta_1 = \\
&= e^{\epsilon}P\left(X \in B_v\right) +2d_K\left(\frac{X}{\sigma},Z\right)\left(1+e^{\epsilon}\right) + \delta_1.
\end{align*}
We already obtained $\epsilon$. We also want to find an upper bound for  $\delta = 2d_K\left(\frac{X}{\sigma},Z\right)\left(1+e^{\epsilon}\right) + \delta_2$. For this purpose we  can use previously shown inequalities concerning Kolmogorov and Wasserstein distance
\begin{align*}
\delta &= 2d_K\left(\frac{X}{\sigma},Z\right)\left(1+e^{\epsilon}\right) + \delta_1 \leqslant \\
&\leqslant 2(1+e^{\epsilon}) \left(\frac{2}{\pi}\right)^{\frac{1}{4}}\sqrt{d_W\left(\frac{X}{\sigma},Z\right)} + \frac{4}{5\sqrt{n}} \leqslant \\
&\leqslant c(\epsilon)\sqrt{\frac{D^2}{\sigma^3}\sum_{i=1}^n{E|X^*_i|^3} + \frac{D^\frac{3}{2} \sqrt{26}}{\sigma^2\sqrt{\pi}} \sqrt{\sum_{i=1}^n{E\left(X^*_i\right)^4}}} + \frac{4}{5\sqrt{n}},
\end{align*}
where
$$
c(\epsilon) = 2(1+e^{\epsilon}) \left(\frac{2}{\pi}\right)^{\frac{1}{4}}.
$$
Summing it up we obtain
\begin{align*}
P(X \in B_u) &\leqslant e^{\epsilon}P(X \in B_v) + 2d_K\left(\frac{X}{\sigma},Z\right)\left(1+e^{\epsilon}\right) + \delta_1 \leqslant \\
&\leqslant e^{\epsilon}P(X \in B_v) + \delta,  
\end{align*}
which concludes the proof. 
\end{proof}

\section{Comparison to standard \\ differential privacy}\label{sect:comparison}

Clearly noiseless privacy is an extension of the regular differential privacy from~\cite{dW1} that is 
applicable to the case when we can assume that the  observer/attacker may treat the raw data of users 
(before being processed)  as random variables. In particular if we assume that all data items are
concentrated in single points (i.e, $P(X_i =x_i)=1$ for all $i$ ) we get the original $(\varepsilon,\delta)$-differential privacy. 

While the standard differential privacy definition guarantees immunity against attacks based on \textit{auxiliary information}
(i.e., from publicly available datasets or even personal knowledge about an individual participating in the protocol), the noiseless privacy is more general as we can either assume that the adversary has no auxiliary information, or assume that there is an upper bound on the size of subset of database entries about which he has some external knowledge. Note that if we assume full auxiliary information, this renders noiseless privacy completely unacceptable, which is very intuitive, as the whole notion of adversarial uncertainty demands that the adversary does not have full knowledge. Moreover, it is often quite too pessimistic to assume that the adversary knows everything except for the single data record which privacy he wants to breach.

\begin{remark}\label{sect:rem1}
See that in the standard differential privacy definition (e.g. \cite{DworkAlgo}) we essentially want
$$
P(M(X)\in B|X = x) \leqslant e^{\epsilon} P(M(X')\in B|X'=x') + \delta,
$$
where $x$ and $x'$ are adjacent, deterministic vectors. \\
This captures the notion of neighboring databases. Our approach is indeed a relaxation of that definition, as we do not necessarily condition the data to have some fixed, deterministic value. We rather treat the data inputs as random variables. In particular, if we have $X = x$ with probability $1$ then our model collapses to standard differential privacy.
\end{remark}
\textit{Differential privacy} has some very useful properties. First of all, it is immune to post-processing, so the adversary cannot get any additional information, and consequently cannot increase the privacy loss by convoluting the result of a mechanism with some deterministic function.
\begin{fact}
Noiseless privacy is, similarly to standard \textit{differential privacy} as stated in \cite{DworkAlgo}, resilient to post-processing. The proof goes almost exactly the same way as for standard differential privacy. Let $f: R \rightarrow R'$ be a deterministic function. Let also $T = \{r \in R: f(r) \in S\}$. Now fix $S \subset R'$, privacy mechanism $M$ and a random vector $X$. We have
\begin{align*}
&P(f(M(X))\in S) = P(M(X) \in T) \leqslant \\
&\leqslant e^{\epsilon}P(M(X') \in T) + \delta = e^{\epsilon}P(f(M(X')) \in S) + \delta,
\end{align*}
which completes the proof of this remark.
\end{fact}

Another important property of differential privacy is its composability. There has been an extended discussion concerning composability of noiseless privacy and its derivatives in \cite{prevDP1,prevDP2,prevDP3}.
